\documentclass{elsarticle}

\usepackage{epsfig}
\usepackage[T1]{fontenc}
\usepackage[latin1]{inputenc}
\usepackage[english]{babel}
\usepackage{times}
\usepackage{color}
\usepackage{subcaption}
\usepackage{amsthm,amsmath,enumerate}
\usepackage{amssymb}
\usepackage{amscd}
\usepackage{latexsym}
\usepackage{tabularx}
\usepackage{stmaryrd}
\usepackage{vmargin}
\usepackage{euscript}
\usepackage{xcolor}
\usepackage{a4}
\usepackage[colorlinks=true,citecolor=black,linkcolor=black,urlcolor=blue]{hyperref}

\newcommand{\arxiv}[1]{\href{http://arxiv.org/abs/#1}{\texttt{arXiv:#1}}}

\setmarginsrb{3.5cm}{2cm}{3.5cm}{2cm}{1cm}{1cm}{1cm}{1cm}

\newtheorem{theorem}{Theorem}

\newtheorem{lemma}[theorem]{Lemma}

\def\paren#1{\left( #1 \right)}
\def\acc#1{\left\{ #1 \right\}}

\def\ceil#1{\left\lceil #1 \right\rceil}

\renewcommand{\le}{\leqslant}
\renewcommand{\ge}{\geqslant}

\begin{document}
\begin{frontmatter}
\title{The complexity of partitioning into disjoint cliques and a triangle-free graph\tnoteref{t1}}
\tnotetext[t1]{To Nina (Pascal's daughter, born on Nov 1 2015) and Ilyas (Marin's son, born on Nov 6 2015)}
\author[LIRMM]{Marin Bougeret}
\author[LIRMM]{Pascal Ochem}
\address[LIRMM]{CNRS - LIRMM, Montpellier, France}

\begin{abstract}
Motivated by Chudnovsky's structure theorem of bull-free graphs,
Abu-Khzam, Feghali, and M\"uller have recently proved that deciding if a graph has a vertex partition
into disjoint cliques and a triangle-free graph is NP-complete for five graph classes.
The problem is trivial for the intersection of these five classes.
We prove that the problem is NP-complete for the intersection of two subsets of size four among the five classes.
We also show NP-completeness for other small classes, such as graphs with maximum degree 4 and line graphs.
\end{abstract}

\begin{keyword}
Graph coloring, NP-completeness.
\end{keyword}

\end{frontmatter}

\section{Introduction}
In this paper we consider the problem of recognizing graphs having a
vertex partition into disjoint cliques and a triangle-free graph.
We say that a graph is \emph{partitionable} if it has such a partition.
The vertices in the $P_3$-free part are colored blue and
the vertices in the $K_3$-free part are colored red. 
This problem is known to be NP-complete on general graphs~\cite{farrugia2004vertex}.
The NP-completeness on bull-free graphs was motivated by an open question
in~\cite{thomasse2013parameterized} (after Thm 2.1) about
the complexity of recognizing the class $\tau_1$ introduced by Chudnovsky~\cite{chudnovsky2012structure}
in her characterization of bull-free graphs.
Abu-Khzam, Feghali, and M\"uller~\cite{AFM15} have then investigated the complexity of deciding
whether a bull-free graph is partitionable. They have shown the following.
\begin{theorem}~\cite{AFM15}
\label{previous}
Recognizing partitionable graphs is NP-complete even when restricted to the following classes:
\begin{enumerate}[(1)]
 \item planar graphs,\label{n1}
 \item $K_4$-free graphs,\label{n2}
 \item bull-free graphs,\label{n3}
 \item $(C_5,\dots,C_t)$-free graphs (for any fixed $t$),\label{n4}
 \item perfect graphs.\label{n5}
\end{enumerate}
\end{theorem}



In Section~\ref{sec:hard}, we prove Theorem~\ref{main} which improves Theorem~\ref{previous}.
The classes~$h_1$ and~$h_2$ of theorem~\ref{main} show that the problem remains
NP-complete for the intersection of two subsets of size four among the five classes of
Theorem~\ref{previous} (graphs in the intersection of the five classes are partitionable).
The class~$h_1$ also answers the open question~\cite{AFM15} of the complexity of recognizing partitionable
Meyniel graphs, since every graph in~$h_1$ is a parity graph and parity graphs correspond to gem-free Meyniel graphs.
We also show NP-completeness for several other classes.
The clases~$h_3$ to~$h_9$ are motivated by the introduction of other natural forbidden
inducesd subgraphs (mainly $C_4$, $K_4^-$, and $K_{1,3}$) and/or restriction on the maximum degree.
The interesting feature of every class is briefly discussed at the end of its dedicated subsection.
We study the tightness of this result in Section~\ref{sec:poly} by considering
all the intersections between every two graph classes of Theorem~\ref{main}.

We use standard notations for graphs (see~\cite{isgci}), some of them are reminded in Figure~\ref{smallGraphs}.
\begin{figure}[htbp]
\begin{center}
\includegraphics[width=130mm]{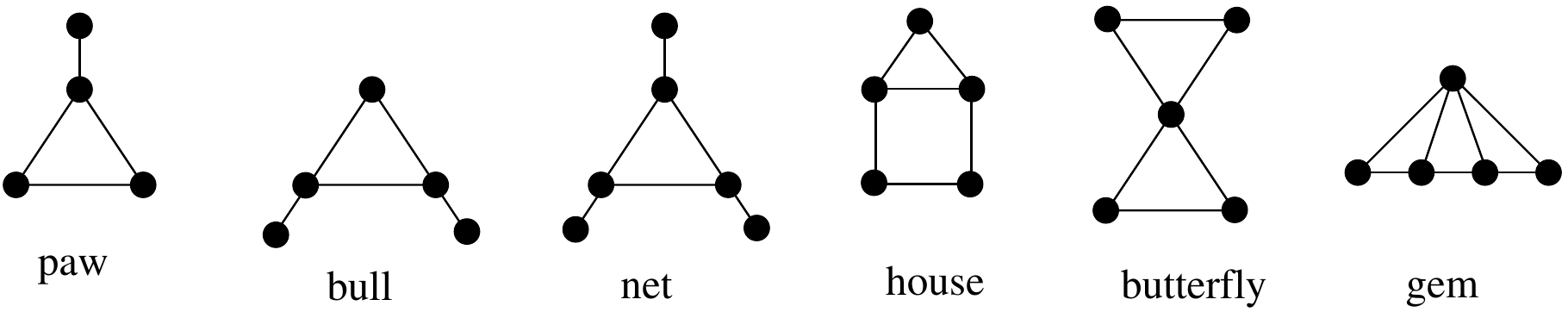}
\caption{Some small graphs and their name.}
\label{smallGraphs}
\end{center}
\end{figure}

A \emph{$k$-vertex} is a vertex of degree $k$.
Given a graph $G$, we denote its line graph by $L(G)$.
Given a graph class $\mathcal{C}$, we denote by $L(\mathcal{C})$
the set of line graphs of graphs in $\mathcal{C}$.


\section{Main result}\label{sec:hard}
In this section we prove the following result.
\begin{theorem}
\label{main}{\ }
Recognizing partitionable graphs is NP-complete even when restricted to the following classes:
\begin{enumerate}[$h_1$:]
 \item planar ($C_4,\dots,C_t$, bull, gem, odd hole)-free graphs with maximum degree 8,\label{a}
 \item planar ($K_4$, bull, house, $C_5,\dots,C_t$)-free graphs,\label{b}
 \item planar ($K_4,C_4$, gem, $C_7,\dots,C_t$, odd hole of length $\ge 7$)-free graphs with maximum degree 7,\label{g}
 \item ($K_4,C_5,\dots,C_t$, net, odd hole)-free graphs with maximum degree 8,\label{h}
 \item ($K_4^-$, butterfly, $C_6,\dots,C_t$)-free graphs with maximum degree 4,\label{c}
 \item ($K_4,K_4^-,C_4,\dots,C_t$, butterfly)-free graphs,\label{d}
 \item planar ($K_{1,3},K_4^-,C_4,\dots,C_t$, odd hole)-free graphs with maximum degree 6,\label{e}
 \item planar ($K_{1,3},K_4^-,C_9,\dots,C_t$, odd hole)-free graphs with maximum degree 5,\label{f}
 \item ($K_{1,3},K_4^-,C_4,\dots,C_t,K_5$, odd hole)-free graphs with maximum degree 5,\label{i}
\end{enumerate}
\end{theorem}

Kratochv\'il proved that \textsc{planar ($3,{}^\le 4$)-sat} is NP-complete~\cite{Kra94}.
In this restricted version of \textsc{sat}, the graph of variable-clause incidences
of the input formula must be planar, every clause is a disjonction of exactly
three literals, and every variable occurs in at most four clauses.
For every class considered in Theorem~\ref{main}, we provide a reduction from \textsc{planar ($3,{}^\le 4$)-sat}.
Given an instance formula $I$ of \textsc{planar ($3,{}^\le 4$)-sat}, we construct a graph $G$
such that $G$ is partitionable if and only if $I$ is satisfiable.

For the classes~$h_1$,~$h_2$,~$h_3$, and~$h_4$,
the boolean value true is associated to the color red,
the boolean value false is associated to the color blue, and the clause gadget is $P_3$.
This way, an unsatisfied clause corresponds to a blue $P_3$.
For the classes~$h_5$ and~$h_6$, the boolean value true is associated to the color blue,
the boolean value false is associated to the color red, and the clause gadget is $K_3$.
This way, an unsatisfied clause corresponds to a red $K_3$.
For brevity, we say that a vertex with the color associated to the boolean value true
(resp. false) is colored true (resp. false).

Given a variable $x$, a \emph{variable gadget} is a graph $G_x$ with
two disjoint subsets of vertices $S_x$ and $S_{\overline{x}}$ such that:
\begin{itemize}
\item There exists an involutive automorphism of $G_x$ which swaps $S_x$ and $S_{\overline{x}}$.
\item There exists a partition of $G_x$ such that every vertex in $S_x$ is colored true
and no blue vertex in $S_x\cup S_{\overline{x}}$ is adjacent to a blue vertex.
\item No partition of $G_x$ is such that both a vertex in $S_x$ and a vertex in $S_{\overline{x}}$
are colored true.
\end{itemize}

The variable gadget depends on the considered graph class and is built on \emph{forcers}.
A forcer is a partitionable graph with a specified vertex $q$.
\begin{itemize}
 \item A red forcer is such that $q$ is red in every partition.
 \item A blue forcer is such that $q$ is blue in every partition and there exists a partition
 such that every neighbor of $q$ is red.
\end{itemize}

We construct $G$ from $I$ as follows.
We take one copy of the variable gadget per variable.
We take one copy of the clause gadget (either $P_3$ or $K_3$) per clause.
Each of the 3 vertices of the clause gadget corresponds to a literal of the clause.
The vertices in $S_x\cup S_{\overline{x}}$ are depicted in green in the representation
of the variable gadgets (Fig.~\ref{var_abgh} and~\ref{var_cd}).
A subset of these green vertices corresponds to the literals of the variable $x$.
For every literal $\ell_x$ of $I$, one vertex corresponding to $\ell_x$
in $G_x$ is identified to the vertex corresponding to $\ell_x$ in the clause gadget.

For the classes~$h_1$,~$h_2$,~$h_3$, and~$h_4$,
we use a \emph{parity labelling} to ensure that $G$ has no induced odd hole.
This labelling assigns a value in $\acc{1,2}$ to every vertex of a subgraph of $G$,
such that the values alternate on every labelled induced path with at least 3 vertices.
In every clause gadget $P_3$ of $G$, the two extremities of $P_3$ are labelled $1$ and the middle
vertex is labelled $2$.
Notice that in the variable gadget in Fig.~\ref{var_abgh}, there exist green vertices
labelled $1$ and $2$ both in $S_x$ and $S_{\overline{x}}$.
Thus, we can make sure that every vertex of a clause gadget is identified to a green vertex corresponding
to the suitable literal having the same parity label. To check that $G$ contains no odd hole,
we consider the 2-connected components of $G$. They have a bounded number of vertices, except
one large 2-connected component which is entirely labelled and contains every clause gadget.

Only a part of the variable gadget is depicted in the figures. The actual size of the variable gadget
depends (linearly) on $t$. The variable gadget consists in sufficiently many copies of the depicted part
that are arranged circularly, as suggested by the dotted edge.
Thus, we can ensure that $G$ is ($C_i,\dots,C_t$)-free by identifying every green vertex to at most one
vertex of a clause a gadget and by requiring that the distance in the variable gadget between two "used" green
vertices is at least $t$. For subclasses of planar graphs, we also make sure to identify green vertices
and vertices of the clause gadget so that $G$ is planar.

Let us prove that this construction provides a reduction.
\begin{lemma}
$I$ is satisfiable if and only if $G$ is partitionable.
\end{lemma}

\begin{proof}
Suppose that $I$ is satisfiable. For every variable $x$, if $x$ is set to true
(resp. false) in the satisfying assignment, then we partition $G_x$ such that
every vertex in $S_x$ (resp. $S_{\overline{x}}$) is colored true
and no blue vertex in $S_x\cup S_{\overline{x}}$ is adjacent to a blue vertex.
This implies that every clause gadget of $G$ contains a vertex colored true.
Moreover, there exists no blue $P_3$ across the intersection of variable gadget and an edge gadget.
So $G$ is partitionable.
Conversely, suppose that $G$ is partitionable and consider a partition of the vertices of $G$.
Thus, every clause gadget of $G$ contains a vertex colored true.
By the property of the variable gadget, if a vertex corresponding to a literal is colored true,
then every vertex corresponding to the opposite literal of the variable is colored false.
Every variable $x$ of $I$ is set to true if and only if there exists a vertex colored true in $S_x$.
By previous discussion, this gives a satisfying assignment of $I$.
\end{proof}

We leave to the reader to check that the variable gadgets satisfy the required properties and
that $G$ contains none of the induced subgraphs that the considered class forbids.

\subsection{Class $h_1$}
We use the blue forcer in Figure~\ref{blue_a} and the variable gadget in Figure~\ref{var_abgh}.
The class~$h_1$ is a subclass of parity graphs.

\begin{figure}[htbp]
\begin{subfigure}{.4\linewidth}
\centering
\includegraphics[width=55mm]{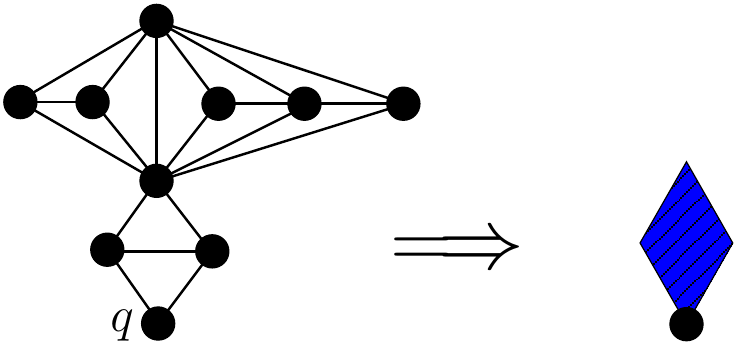}
\caption{The blue forcer for~$h_1$.}
\label{blue_a} 
\end{subfigure}\hfill
\begin{subfigure}{.5\linewidth}
\centering
\includegraphics[width=70mm]{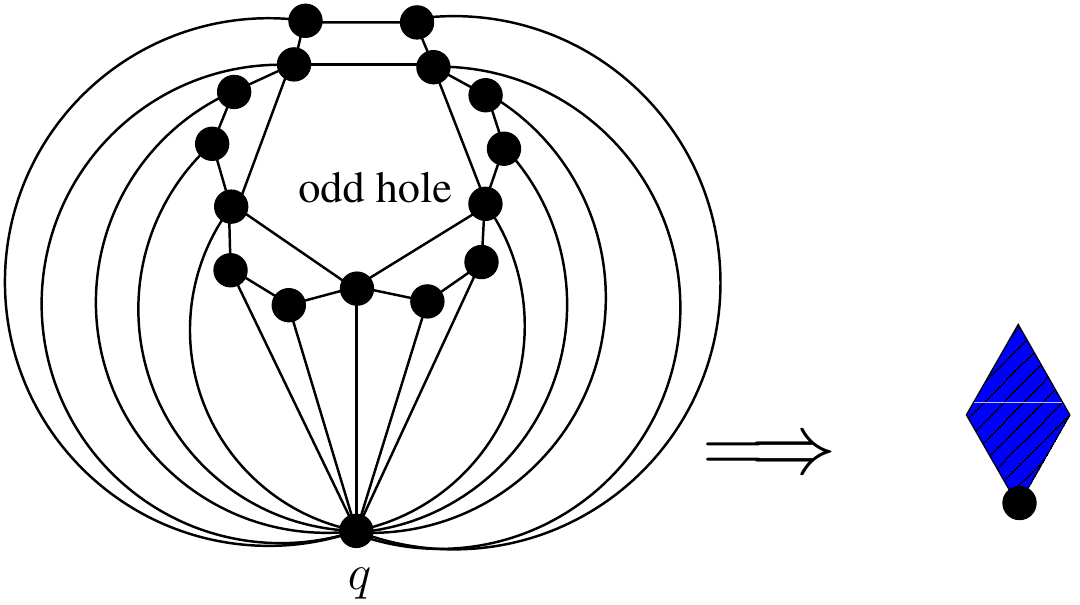}
\caption{The blue forcer for~$h_2$.}
\label{blue_b} 
\end{subfigure}
\caption{}
\end{figure}

\begin{figure}[htbp]
\begin{center}
\includegraphics[width=125mm]{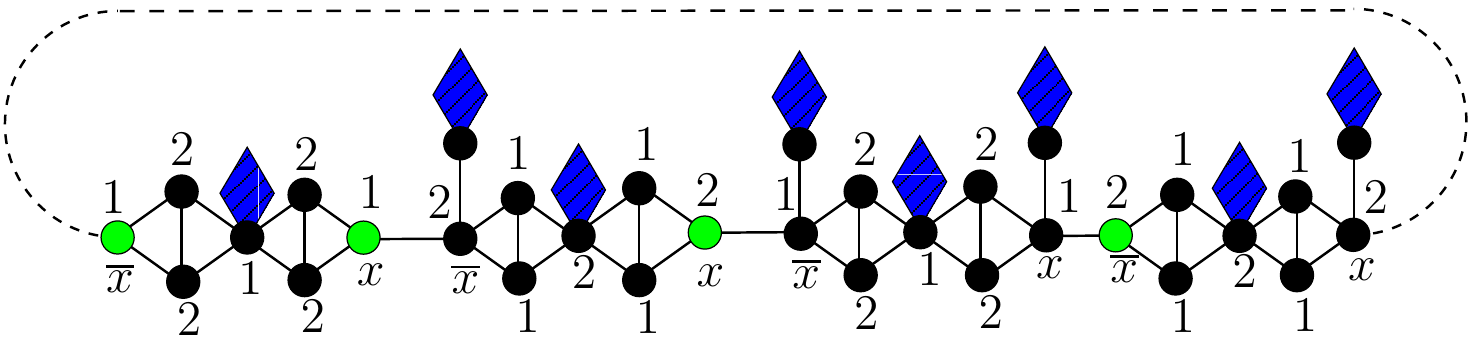}
\caption{The variable gadget for~$h_1$,~$h_2$,~$h_3$, and~$h_4$. The corresponding clause gadget is $P_3$.}
\label{var_abgh}
\end{center}
\end{figure}

\subsection{Class $h_2$}
We use the blue forcer in Figure~\ref{blue_b} and the variable gadget in Figure~\ref{var_abgh}.
We make sure that the length of the odd-hole in the blue forcer is at least $t+1$,
so that $G$ is ($C_5,\dots,C_t$)-free.
Notice that the maximum degree of $G$ is a linear function of $t$.
The class~$h_2$ is almost perfect and $K_4$-free (which would imply partitionable),
except that it contains large odd holes.

\subsection{Class $h_3$}
The graph $F$ on the left of Figure~\ref{red_g} admits no partition such that every
neighbor of $w$ is red. However, it has a partition such that $w$ is blue and a partition such that $w$ is red.
We use five copies of $F$ to obtain the red forcer, as depicted on the right of Figure~\ref{red_g}.
Suppose for contradiction that the red forcer admits a partition such that the specified vertex $q$ is blue.
By the property of $F$, at least one neighbor $r$ of $q$ is also blue.
Again, by the property of $F$, at least one of neighbor $s$ of $r$ in the copy of $F$ attached to $r$ is also blue.
Then $qrs$ is a blue induced $P_3$, which is a contradiction.

We obtain the blue forcer from the red forcer using Figure~\ref{red2blue}.
We use the variable gadget in Figure~\ref{var_abgh}.
Notice that $G$ only avoids odd holes of length at least 7 since the red forcer contains~$C_5$.
The class~$h_3$ is almost perfect and $K_4$-free (which would imply partitionable),
except that it contains~$C_5$.

\begin{figure}[htbp]
\begin{subfigure}{.7\linewidth}
\centering
\includegraphics[width=75mm]{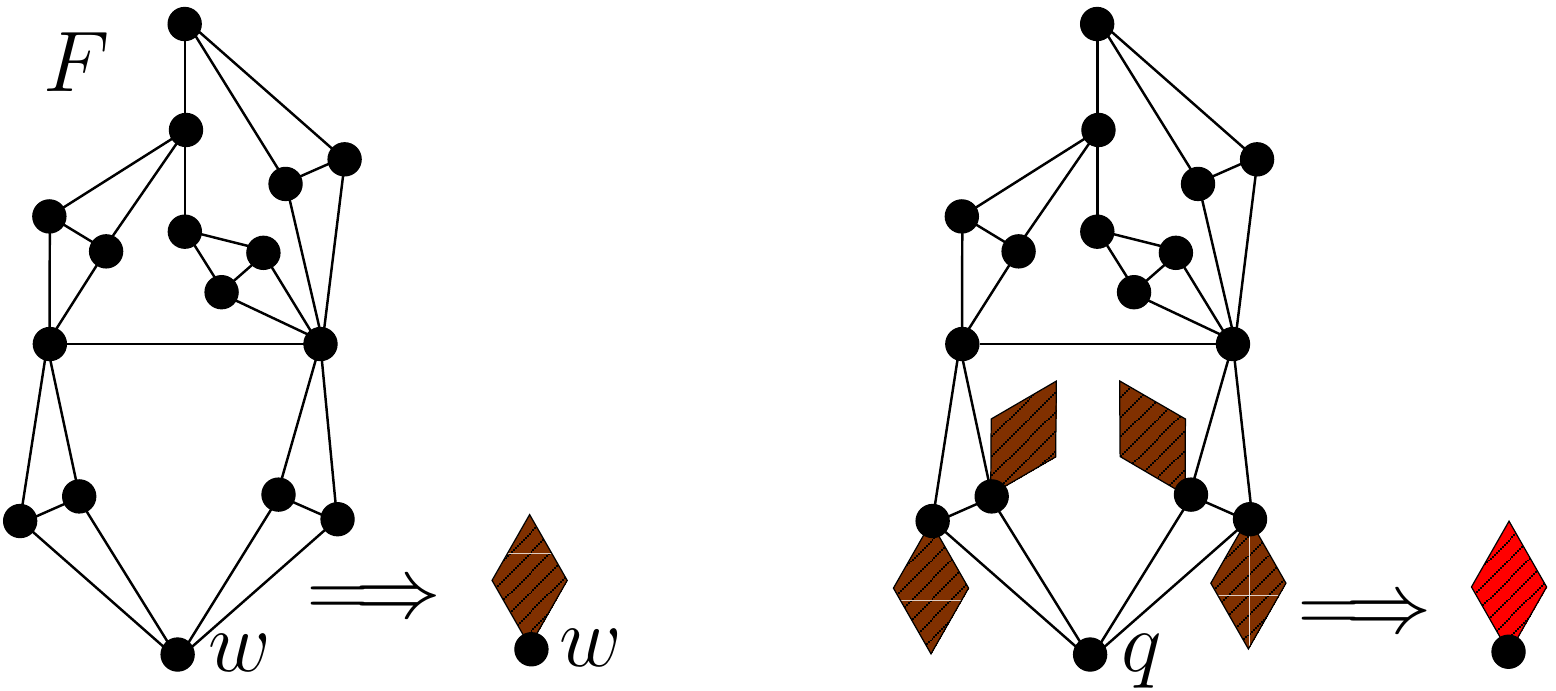}
\caption{The red forcer for~$h_3$.}
\label{red_g} 
\end{subfigure}\hfill
\begin{subfigure}{.3\linewidth}
\centering
\includegraphics[width=30mm]{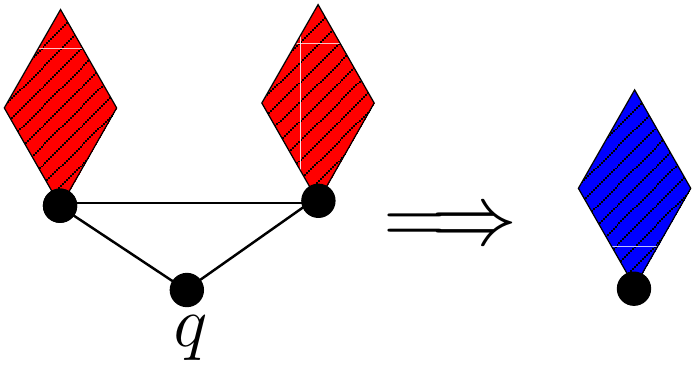}
\caption{Obtaining a blue forcer from a red forcer for~$h_3$,~$h_5$,~$h_6$,~$h_7$, and~$h_8$.}
\label{red2blue} 
\end{subfigure}\hfill
\caption{}
\end{figure}

\subsection{Class $h_4$}
We describe a blue forcer for~$h_4$.
Consider the graph $\overline{C_7}$ with vertices $v_0,\ldots,v_6$ such that $v_i$
is adjacent to $v_j$ unless $|i-j|\le1\pmod{7}$. Every partition of $\overline{C_7}$
is such that three consecutive vertices are blue and the other vertices are red.
So a partition is characterized by its monochromatic blue edge $v_tv_{t+2}$
(indices are taken modulo 7).
We add to $\overline{C_7}$ a false twin $v'_i$ to $v_i$ for $i\in\acc{3,4,6}$.
A vertex contained in a monochromatic blue edge cannot have a false, since it would create
a blue $P_3$. Thus, the monochromatic blue edge is either $v_5v_0$ or $v_0v_2$.
The blue forcer is obtained by adding a copy of $K_4^-$ and identifying a 2-vertex of $K_4^-$
with $v_0$. The specified vertex of the blue forcer is the 2-vertex
(the "other" 2-vertex of $K_4^-$).
Then we use the variable gadget in Figure~\ref{var_abgh}.
The class~$h_4$ is almost perfect and $K_4$-free (which would imply partitionable),
except that it contains $\overline{C_7}$.

\subsection{Class $h_5$}
We use the red forcer in Figure~\ref{red_c}.
We obtain a blue forcer for~$h_5$ from this red forcer using the construction of Figure~\ref{red2blue}.
We use the variable gadget in Figure~\ref{var_cd}.
Notice that degree 4 is best possible since graphs with maximum degree 3 are partitionable, as shown in the last section. 

\begin{figure}[htbp]
\begin{subfigure}{.5\linewidth}
\centering
\includegraphics[width=50mm]{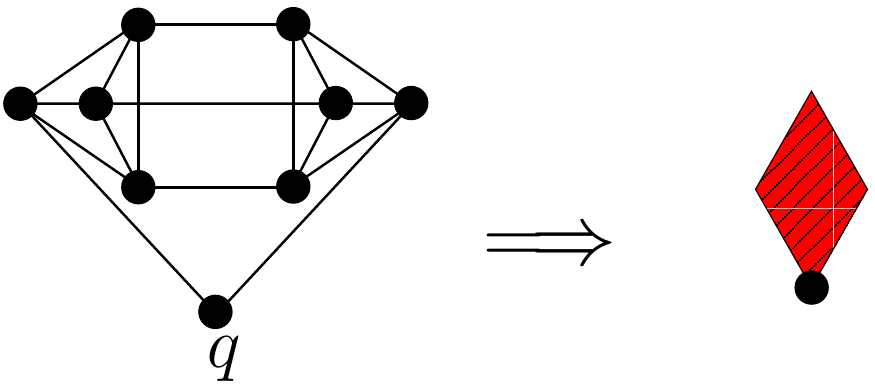}
\caption{The red forcer for~$h_5$.}
\label{red_c} 
\end{subfigure}\hfill
\begin{subfigure}{.5\linewidth}
\centering
\includegraphics[width=30mm]{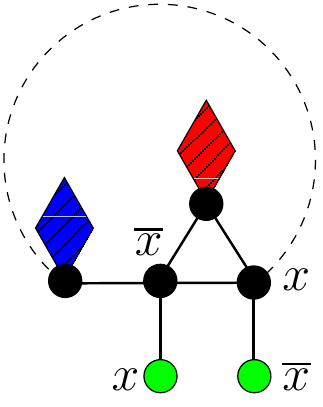}
\caption{The variable gadget for~$h_5$ and~$h_6$. The corresponding clause gadget is $K_3$.}
\label{var_cd}
\end{subfigure}
\caption{}
\end{figure}


\subsection{Class $h_6$}
To obtain a red forcer for~$h_6$, we first prove that $h_6$ is not a subclass of partitionable graphs.
By probabilistic arguments, there exists a graph $J_1$ among the random graphs $G_{n,p}$ with probability $p=n^{-1+1/2t}$ such that:
\begin{itemize}
 \item $J_1$ contains at most $\tfrac{n}{100}$ cycles of length at most $2t$.
 \item $J_1$ contains a path of length at least $\tfrac{99}{100}n$ (Theorem 8.1 in~\cite{random}).
 \item $J_1$ contains no independent set of size $\tfrac{n}7$.
\end{itemize}
We call \emph{bag} a subset of three vertices that induce either a $P_3$ or a $K_3$.
Let $J_2$ be the graph induced by the $\tfrac{99}{100}n$ of the mentioned path in $J_1$.
Thus, $J_2$ can be split into $\tfrac{33}{100}n$ bags.
Notice that we can destroy every cycle of length at most $2t$ in $J_1$ by removing at most $\tfrac{n}{100}$ vertices.
For every such vertex $v$, we remove from $J_2$ the bag containing $v$.
This way, we obtain a graph $J_3$ with girth at least $2t+1$
and having at least $\tfrac{33}{100}n-\tfrac{n}{100}=\tfrac8{25}n$ bags.
We obtain the graph $J_3$ by adding to $J_2$ the edge between the extremities
of each of the $\tfrac8{25}n$ copies of $P_3$ contained in a bag.
Every vertex in $J_3$ is contained in exactly one triangle.
Also, every cycle of length at least 4 in $J_3$ has length at least $\tfrac23(2t+1)\ge t+1$.
This implies that $J_3\in h_6$.
Notice that neither $J_1$ nor $J_3$ contains an independent set of size $\tfrac{n}7$.
Suppose for contradiction that $J_3$ is partitionable.
Without loss of generality, every triangle contains exactly one blue vertex.
So, the graph induced by the blue vertices is $(P_3,K_3)$-free and is thus bipartite.
Thus, there exists an induced blue independent set of size at least $\tfrac12\times\tfrac8{25}n>\tfrac{n}7$.
This contradiction shows that $J_3$ is a graph in~$h_6$ that is not partitionable.

Consider a graph $J_{min}$ in~$h_6$ that is not partitionable and is minimal with respect to the
number of edges. The red forcer is obtained from $J_{min}$ by subdividing once an edge $uv$ that is not contained in a triangle.
The specified vertex is the subdivision vertex.
By minimality, the graph $J_{min}\setminus\acc{uv}$ is partitionable and both $u$ and $v$ are blue in every partition.
Then the specified vertex must be red in order to avoid a blue $P_3$.

We obtain a blue forcer for~$h_6$ from this red forcer using the construction of Figure~\ref{red2blue}.
We use the variable gadget in Figure~\ref{var_cd}.

The class~$h_6$ is interesting because it avoids cycles of length $4$ to $t$ as (not necessarily induced) subgraphs.
Moreover, every vertex is contained in at most one triangle since it is butterfly-free.

\subsection{Classes $h_7$, $h_8$, and $h_9$}
The classes~$h_7$,~$h_8$, and~$h_9$ are subclasses of
($K_{1,3},K_4^-,K_5$, odd hole)-free graphs, which correspond
to the class of line graphs of bipartite graphs with maximum degree 4.
For convenience, we thus consider the corresponding edge-partitioning problem described as follows.
A triangle-free graph is \emph{edge-partitionable} in red and blue if
the blue edges induce a star forest
and every vertex is incident to at most two red edges.
We then have that a triangle-free graph $G$ is edge-partitionable if and only if $L(G)$ is partitionable.

\begin{lemma}
\label{42}
If a graph $G$ is such that every edge is incident to a 4-vertex and a 2-vertex,
then $G$ is edge-partitionable and every edge-partition of $G$ is such that every 4-vertex
is incident to exactly two blue edges and every 2-vertex is incident to exactly one blue edge.
\end{lemma}

\begin{proof}
In every edge-partition of $G$, every 4-vertex is incident to at least two blue edges.
So, at least half of the edges of $G$ are blue. Moreover, every 2-vertex is incident
to at most one blue edge, since otherwise there would be a blue path on 4 edges.
So, at most half of the edges of $G$ are blue.
Thus, there are equally many red and blue edges in every edge-partition of $G$.
Moreover, every 4-vertex is incident to exactly two blue edges and every 2-vertex is incident to exactly one blue edge.

Now we show that $G$ is edge-partitionable.
We define the contraction of a 2-vertex $v$ adjacent to $u_1$ and $u_2$
as the deletion of $v$ and the addition of one (additional) edge $u_1u_2$.
Let $G'$ be the multigraph obtained by contracting every 2-vertex of $G$.
Since $G'$ is 4-regular, we can orient the egdes of $G'$ such that the out-degree of every vertex is 2.
We extend this orientation of $G'$ to $G$ such that the incidences of the arcs to the 4-vertices are unchanged.
We assign red (resp. blue) to an edge of $G$ if its tail (resp. head)
is incident to a 4-vertex. This gives a valid egde-partition of $G$ since every vertex
is incident to at most two red edges and the graph induced by the blue edges
is a star forest such that the 4-vertices are the centers of the stars.
\end{proof}

\begin{lemma}
\label{edge}
Recognizing edge-partitionable bipartite graphs with maximum degree 4 is NP-complete even when restricted to:
\begin{itemize}
 \item[$h'_7$:] planar ($C_3,\dots,C_t$)-free graphs such that every 4-vertex is a cut vertex.
 \item[$h'_8$:] planar ($C_9,\dots,C_t$)-free graphs such that no edge is incident to two 4-vertices
 and such that every 4-vertex is a cut vertex.
 \item[$h'_9$:] ($C_3,\dots,C_t$)-free graphs such that no edge is incident to two 4-vertices.
\end{itemize}
\end{lemma}

\begin{proof}
For every class, we give a red edge forcer.
The red edge forcers for $h'_7$ and $h'_8$ are depicted in Figure~\ref{red_gghh}.
Let us construct a red edge forcer for $h'_9$. Consider a cage $J$ with degree 4 and girth $\ceil{\tfrac{t+1}2}$
(see~\cite{EJ} for more information on cages) and subdivide every edge once to obtain the bipartite graph~$J'$.
We obtain a red edge forcer for $h'_9$ by adding a 1-vertex adjacent to one of the 2-vertices of $J'$.
By Lemma~\ref{42}, the added edge is necessarily red.

\begin{figure}[htbp]
\begin{subfigure}{.4\linewidth}
\centering
\includegraphics[width=40mm]{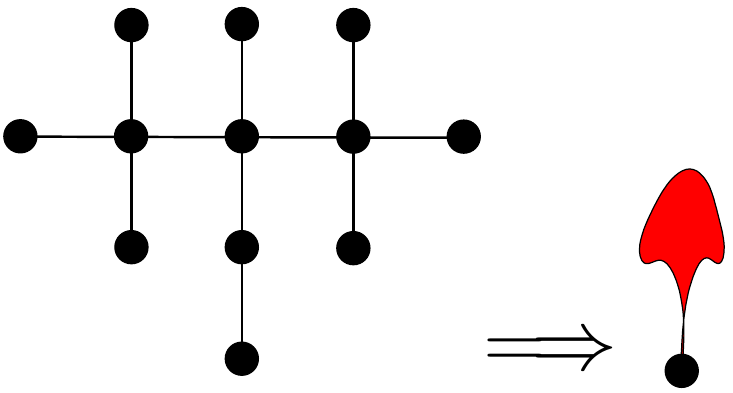}
\caption{The red edge forcer for $h'_7$.}
\label{red_gg}
\end{subfigure}\hfill
\begin{subfigure}{.6\linewidth}
\centering
\includegraphics[width=40mm]{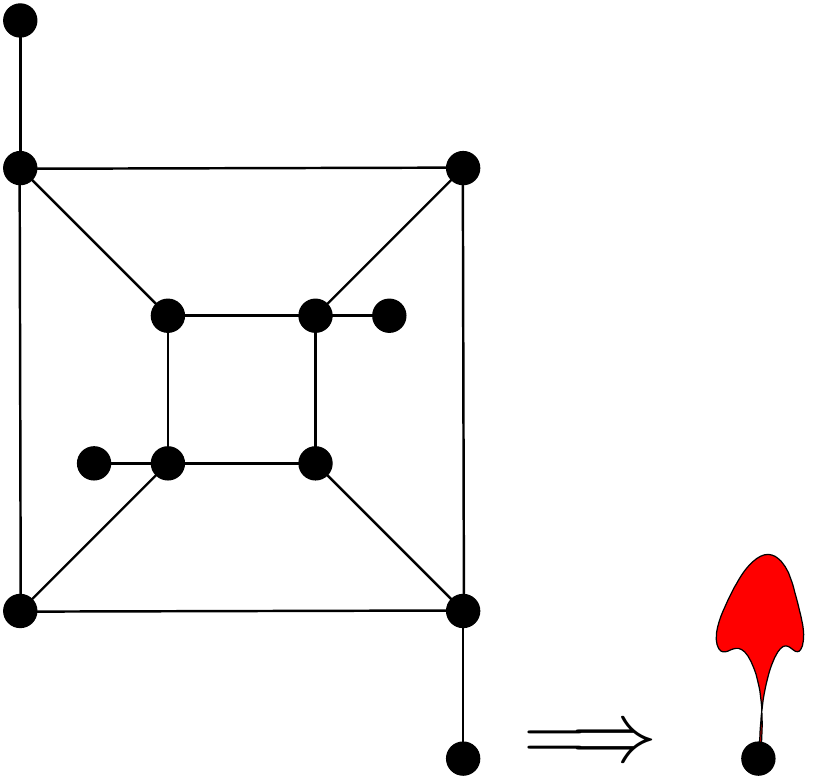}
\caption{The red edge forcer for $h'_8$.}
\label{red_hh}
\end{subfigure}
\caption{}
\label{red_gghh}
\end{figure}

For every class, we obtain a blue edge forcer from two copies of the red edge forcer of the class,
as depicted in Figure~\ref{r2bl}.
This is the counterpart for edge-partition of the construction in Figure~\ref{red2blue}.
These red and blue edge forcers are used in the construction of the variable gadget (Fig.~\ref{var_l})
and the clause gadget (Fig.~\ref{clause_l}).
To obtain the instance graph of the edge-partition problem, we identify every green edge corresponding to a
literal in the variable gadget to the green edge corresponding to this literal in the clause gadget.
The boolean value false is associated to the color red, so that the clause gadget is not edge-partitionable
if the clause is not satisfied. The blue forcer in the clause gadget implies that if the edge
corresponding to a literal is colored blue, then this edge is not incident to another blue edge
in the variable gadget. So we can check that if a green edge marked $x$ in the variable gadget is colored blue,
then every edge marked $x$ must be blue and every edge marked $\overline{x}$ must be red. 

\begin{figure}[htbp]
\begin{subfigure}{.3\linewidth}
\centering
\includegraphics[width=30mm]{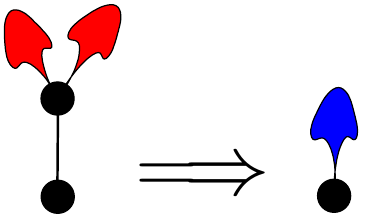}
\caption{Obtaining a blue edge forcer from a red edge forcer.}
\label{r2bl}
\end{subfigure}\hfill
\begin{subfigure}{.3\linewidth}
\centering
\includegraphics[width=50mm]{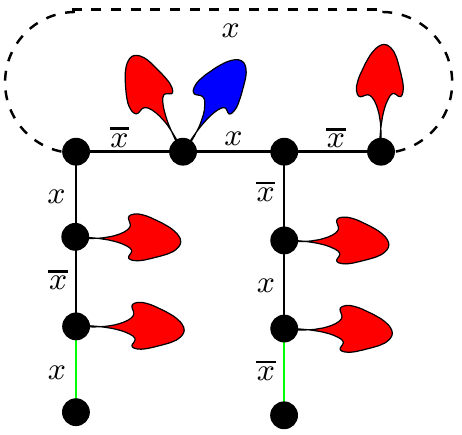}
\caption{The variable gadget.}
\label{var_l}
\end{subfigure}
\begin{subfigure}{.3\linewidth}
\centering
\includegraphics[width=25mm]{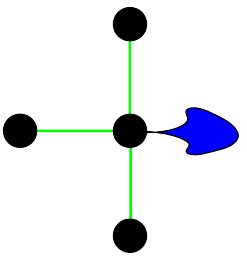}
\caption{The clause gadget.}
\label{clause_l}
\end{subfigure}
\caption{The gadgets for $h'_7$, $h'_8$, and $h'_9$.}
\label{gadget_l}
\end{figure}

\end{proof}

If a planar graph with maximum degree 4 is such that every 4-vertex is a cut vertex,
then its line graph is planar~\cite{Sed63}. We thus have $L(h'_7)\subset h_7$, $L(h'_8)\subset h_8$, and $L(h'_9)\subset h_9$.
So, deciding whether a graph in $h_7$, $h_8$, and $h_9$ is partitionable is NP-complete.

\section{Polynomial classes}\label{sec:poly}
We gather here known graph classes for which deciding whether a member is partitionable
is polynomial time solvable.

\subsection{Chordal graphs}
See~\cite{AFM15}.

\subsection{$P_4$-free graphs}
Abu-Khzam, Feghali, and M\"uller~\cite{AFM15} have shown that we can test in polynomial time
whether a $P_4$-free graph is partitionable.

\subsection{Paw-free graphs}
Olariu~\cite{Ola88} has proved that every connected component of a paw-free graph is either triangle-free or
a complete multipartite graph. Since a complete multipartite graph is $P_4$-free,
we can test in polynomial time whether a paw-free graph is partitionable.

\subsection{$\overline{K_k}$-free graphs}
Recall that the Ramsey number $R(s,t)$ is the least $r$ such that every graph
on $r$ vertices contains an independent set on $s$ vertices or a clique on $t$ vertices.
In a $\overline{K_k}$-free graph, every $K_3$-free induced subgraph thus contains at most $R(k,3)-1$ vertices.
Consider a $\overline{K_k}$-free graph $G$ with $n$ vertices.
For each of the $O\paren{n^{R(k,3)-1}}$ subsets $S$ of at most $R(k,3)-1$ vertices of $G$,
we can test in $O(n^3)$ time whether $S$ induces a $K_3$-free graph and $G\setminus S$ induces a $P_3$-free graph.
So, for any fixed $k$, we can test in polynomial time whether a $\overline{K_k}$-free graph is partitionable.

\section{Partitionable classes}
We gather below known classes of partitionable graphs.
The motivation is the conjecture that for every two of the classes considered in Theorem~\ref{main},
their intersection is a subclass of partitionable graphs. For most pair of classes,
we have identified a class of partitionable graphs containing the intersection,
see Table~\ref{tab:intersection}. The remaining two open cases of the conjecture are denoted by a question mark.

\begin{table}[htbp]
\begin{center}
\begin{tabular}{|c|c|c|c|c|c|c|c|c|}
 \hline
 & $h_1$ & $h_2$ & $h_3$ & $h_4$ & $h_5$ & $h_6$ & $h_7$ & $h_8$\\
 \hline
$h_2$ & $p_2$ & & & & & & & \\
 \hline
$h_3$ & $p_2$ & $p_2$ & & & & & & \\
 \hline
$h_4$ & $p_2$ & $p_2$ & $p_2$ & & & & & \\
 \hline
$h_5$ & $p_3$ & $p_3$ & ? & $p_2$ & & & & \\
 \hline
$h_6$ & $p_2\cap p_3$ & $p_3$ & $p_2$ & $p_2\cap p_3$ & ? & & & \\
 \hline
$h_7$ & $p_3$ & $p_2\cap p_3$ & $p_2$ & $p_2\cap p_3$ & $p_4\cap p_5$ & $p_2\cap p_4$ & & \\
 \hline
$h_8$ & $p_3\cap p_5$ & $p_2\cap p_3$ & $p_2\cap p_5$ & $p_2$ & $p_4$ & $p_2\cap p_4\cap p_5$ & $p_5$ & \\
 \hline
$h_9$ & $p_3\cap p_5$ & $p_2\cap p_3\cap p_5$ & $p_2\cap p_5$ & $p_2\cap p_3$ & $p_4$ & $p_2\cap p_4$ & $p_5$ & $p_5$\\
 \hline
\end{tabular}
\caption{Intersections of the classes considered in Theorem~\ref{main}.}
\label{tab:intersection}
\end{center}
\end{table}

A graph $G$ is \emph{$(d_1,\ldots,d_l)$-colorable} if the vertex set of $G$ can be partitioned
into subsets $V_1,\ldots,V_l$ such that the graph induced by the vertices of $V_i$
has maximum degree at most $d_i$ for every $1\le i\le l$. As it is well known,
for every $a,b\ge0$, every graph with maximum degree $a+b+1$ is $(a,b)$-colorable.

\subsection{$p_1$: $(1,0,0)$-colorable graphs}
Every $(1,0,0)$-colorable graph is partitionable since the color class of degree 1 induces a $P_3$-free graph
and the two color classes of degree 0 induce a bipartite graph, which is $K_3$-free.

\subsection{$p_2$: ($K_4,\overline{C_7}$, odd hole)-free graphs}
This class corresponds to perfect graphs with maximum clique size 3.
By the strong perfect graph theorem, their chromatic number is at most 3.
So, they are $(0,0,0)$-colorable and thus $(1,0,0)$-colorable.
To quickly check that an entry in Table~\ref{tab:intersection} is a subclass of $p_2$, recall
that $K_4^-$, $C_4$, house, and gem are subgraphs of $\overline{C_7}$.

\subsection{$p_3$: ($K_4^-$, house, net)-free graphs}
We define a \emph{big clique} as a maximal clique with at least 3 vertices.
Let $G$ be a ($K_4^-$, house, net)-free graph and let $B$ be a big clique in $G$.
Consider a vertex $x\in G\setminus B$. Then $x$ cannot be adjacent to every vertex in $B$ since $B$ is maximal.
Also, $x$ cannot be adjacent to at least two vertices in $B$ since $G$ is $K_4^-$-free.
So every vertex in $G\setminus B$ is adjacent to at most one vertex in $B$.
Since $G$ is house-free, two vertices in $G\setminus B$ that are adjacent to distinct vertices in $B$
must be non-adjacent.
Since $G$ is net-free, at most two vertices in $B$ are adjacent to a vertex in $G\setminus B$.
For every big clique in $G$, we color blue every vertex of the big clique that has no neighbor outside of the big clique.
The remaining vertices of $G$ are colored red. This gives a partition of $G$ since
the blue part is a disjoint union of cliques and the red part is triangle-free.

To quickly check that an entry in Table~\ref{tab:intersection} is a subclass of $p_3$, recall
that $C_4$ is a subgraph of house and that bull is a subgraph of net.

\subsection{$p_4$: ($K_4^-$, $K_{1,3}$, $K_5$, butterfly)-free graphs}
Let us call $p'_4$ the class of $K_3$-free graphs with maximum degree 4
and such that no edge is incident to two vertices with degree at least 3.
We show that every graph in $p'_4$ is edge-partitionable.
Suppose that $G$ is a counterexample to this statement that is minimal with respect to
the number of edges. 
If $G$ contains an edge $e$ adjacent to at most 2 edges, then we can extend
an edge-partition of $G\setminus e$ by assigning red to $e$
unless every edge adjacent to $e$ is red.
So, by minimality, every edge of $G$ is adjacent to at least 3 edges.
We also have that $G$ is in $p'_4$, so that every edge of $G$ is incident
to a vertex with degree 3 or 4 and a vertex with degree 1 or 2.
Notice that $G$ is a subgraph of some graph considered in Lemma~\ref{42}.

So $G$ is edge-partitionable, and thus every graph in $p'_4$ is edge-partitionable.
Since $p_4=L(p'_4)$, every graph in $p_4$ is partitionable.

To quickly check that an entry in Table~\ref{tab:intersection} is a subclass of $p_4$,
recall that planar graphs are $K_5$-free.

\subsection{$p_5$: planar ($K_{1,3},K_4^-,C_4,\dots,C_{10}$)-free graphs with maximum degree 5}
Let us call $p'_5$ the class of planar graphs with girth at least 11, maximum degree 4,
and such that no edge is incident to two 4-vertices.
We show that every graph in $p'_5$ is edge-partitionable.
Suppose that $G$ is a counterexample to this statement that is minimal with respect to
the number of edges. So $G$ is connected and contains at least 4 edges.

Firstly, $G$ does not contain a vertex $v$ with degree at most 3 that is adjacent
to a 1-vertex $w$. We can extend an edge-partition of $G\setminus\acc{w}$
by assigning red to the edge $vw$, unless $v$ is incident to two red edges in $G\setminus\acc{w}$.

Secondly, $G$ does not contain a 4-vertex $v$ adjacent to at least two 1-vertices $w_1$ and $w_2$.
Since $G$ contains no edge incident to two 4-vertices, the other neighbors $w_3$ and $w_4$
of $v$ have degree at most 3.
If there exists a partition of $G\setminus\acc{w_1,w_2}$ such that $vw_3$ and $vw_4$
have the same color, then we assign blue to $vw_1$ and $vw_2$.
Otherwise, we assume without loss of generality
that $vw_3$ is red and $vw_4$ is forced to be blue.
This means that the other two edges incident to $w_4$ are red.
So we can assign blue to $vw_1$ and red to $vw_2$.

By the previous properties, the graph $G'$ obtained from $G$ by removing every 1-vertex of $G$
has minimum degree 2. It is well known that every planar graph with girth at least $5k+1$
and minimum degree 2 contains a path of $k$ 2-vertices (Lemma 5 in~\cite{NSR97}).
Consequently, $G'$ contains two adjacent 2-vertices $x$ and $y$.
By the previous properties, $x$ and $y$ also have degree 2 in $G$.
Then we can extend an edge-partition of $G\setminus xy$ by assigning red to $xy$.
This contradiction shows that every graph in $p'_5$ is edge-partitionable.

Notice that $p_5\subsetneq L(p'_5)$ since the line graph of a graph in $p'_5$ can be non-planar.
So, every graph in $p_5$ is partitionable.

\subsection{$p_6$: graphs with maximum degree 3}
Whereas the problem is NP-complete for graphs with maximum degree 4, graphs with maximum degree 3 
are partitionable since they are $(1,1)$-colorable graphs and thus $(1,0,0)$-colorable.

\end{document}